\documentclass[sigconf,nonacm]{acmart}

\usepackage{multirow}%
\usepackage{amsmath,amsfonts}%
\usepackage{amsthm}%
\usepackage{mathrsfs}%
\usepackage[title]{appendix}%
\usepackage{xcolor}%
\usepackage{textcomp}%
\usepackage{manyfoot}%
\usepackage{booktabs}%
\usepackage{algorithm}%
\usepackage{algorithmicx}%
\usepackage{algpseudocode}%
\usepackage{subcaption}%
\usepackage{stmaryrd}%

\newtheorem{theorem}{Theorem}
\newtheorem{proposition}{Proposition}%
\newtheorem{lemma}{Lemma}%
\newtheorem{assumption}{Assumption}%

\title{A Geometric Decision Procedure for STL Feasibility and Repair}

\author{Avinash Malik}
\email{avinash.malik@auckland.ac.nz}
\affiliation{%
  \department{Department of Electrical, Computer, and Software Engineering}
  \institution{University of Auckland}
  \city{Auckland}
  \country{New Zealand}
}

\begin{document}

\begin{abstract}
  Feasibility of a Signal Temporal Logic (STL) specification asks
  whether, from a given initial state, an admissible control input can
  satisfy its continuous-time requirements under actuator limits.
  Establishing feasibility is essential for detecting unrealizable tasks
  and repairing specifications before execution. State-of-the-art
  Mixed-Integer Linear Programming (MILP) methods face growing model
  sizes as temporal horizons and formula complexity increase, and can
  accept control witnesses that violate continuous-time STL semantics
  between samples. We propose a geometric approach to STL feasibility
  and temporal repair that translates temporal requirements into spatial
  constraints evaluated at the initial state. Its compiled feasibility
  check is independent of horizon length for fixed compiled dimensions.
  We prove nonnegative continuous robustness under a joint controller
  assumption. On 6D UAV benchmarks, median geometric repair times are
  $0.265$--$1.510$ ms; we show up to $1.68\times10^4$ speedup over MILP
  feasibility checker. All geometric runs pass \textit{continuous} STL
  validation, whereas majority of completed MILP control witnesses
  violate continuous-time STL. 
\end{abstract}

\maketitle

\section{Introduction}
\label{sec1}

Signal Temporal Logic (STL)~\cite{maler2004monitoring} is a formal
language used to specify continuous real-time behavior, such as
reach-avoid requirements, deadline constraints, and spatial bounds over
continuous physical states. In formal methods and automated control,
\emph{STL control synthesis} is the problem of automatically
constructing a feedback control policy or state-input trajectory that
guarantees a physical system satisfies a given temporal logic
specification over time~\cite{raman2014synthesis, belta2019formal}. STL
synthesis has become a foundational building block in autonomous
systems, serving as the temporal logic layer for real-time Control
Barrier Function Quadratic Program (CBF-QP)
controllers~\cite{lindemann2018control}, Large Language Model (LLM) task
translation from natural language to formal
specifications~\cite{he2022deepstl, chen2023nl2tl}, autoregressive task
and motion planning~\cite{chen2023autotamp}, and data-driven
specification mining and inference~\cite{bartocci2022survey}.

A central challenge across these domains is handling
\emph{infeasibility}. In physical controller synthesis, infeasibility
occurs when a specification cannot be satisfied due to system dynamics,
actuator limits, or conflicting temporal
constraints~\cite{ghosh2016diagnosis, belta2019formal}. In modern
automated planning pipelines, it frequently arises from LLMs generating
logically incorrect or physically impossible temporal
deadlines~\cite{chen2023autotamp, chen2023nl2tl}. For both automated
verification tools and LLM-in-the-loop agents, simply returning a binary
failure signal is insufficient; the synthesis framework must provide
actionable diagnostic feedback to explain why a task failed and how to
correct it~\cite{kress2018synthesis, ghosh2016diagnosis}. In discrete
reactive synthesis, such feedback is typically delivered as
counterexample trajectories on symbolic
abstractions~\cite{kress2018synthesis}. In continuous STL synthesis, the
standard decision procedure encodes both the temporal logic formula and
the system dynamics into a Mixed-Integer Linear Program
(MILP)~\cite{raman2014synthesis, belta2019formal,ghosh2016diagnosis}.
When the specification is infeasible, existing approaches extract an
Irreducible Infeasible Set (IIS) from the optimization solver to isolate
the subset of conflicting constraints~\cite{ghosh2016diagnosis}.

However, MILP-based synthesis suffers from a fundamental structural
limitation. Standard encodings model time by discretizing the horizon
into explicit temporal steps, introducing binary decision variables for
every predicate at every time step, with nested temporal operators
compounding this variable count multiplicatively~\cite{kurtz2022mixed}.
Furthermore, sampled STL synthesis and
repair~\cite[Sec.~2.1]{ghosh2016diagnosis} can return a control witness
that satisfies every sampled constraint while violating the
specification between samples. Sampled satisfaction alone is therefore
an \textit{unsound} certificate of continuous-time STL satisfaction: the
returned trajectory may have negative robustness despite passing the
discrete check.

Our primary \textbf{contribution} is a geometric decision procedure for
STL synthesis feasibility that addresses scalability and continuous-time
feasibility certification. Rather than constructing a trajectory
optimization problem over discretized time, our procedure evaluates
synthesis feasibility via a local matrix-vector inequality check with no
temporal sampling loop. Certificate evaluation is horizon-independent;
experiments measure compilation and repair costs, including deeply
nested formulae.

\subsection{Motivating Example and Overview of the Proposed Decision
  Procedure}
\label{sec:motivating_example}

We use a robotic reach and dwell example for motivating the need for our
procedure.

\subsubsection{Physical Scenario and STL Specification}
\label{sec:phys-scen-stl}
Consider a robot with position $p(t)\in\mathbb R^2$, velocity
$v(t)\in\mathbb R^2$, and state $x(t)=[p(t)^T,v(t)^T]^T$. It starts at
$p(0)=[20.0,18.0]^T\text{ m}$ with $v(0)=[1.0,1.0]^T\text{ m/s}$ and obeys
independent double-integrator dynamics:
$ \dot p(t)=v(t),\dot v(t)=u(t), \|u(t)\|_\infty\le a_{\max}, \|v(t)\|_\infty\le
v_{\max}$, where $a_{\max}=1.0\text{ m/s}^2$ and
$v_{\max}=2.0\text{ m/s}$.

The task planning layer issues a multi-stage temporal requirement:
\emph{``Within $2$ seconds, reach a designated 2D goal region located at
  $p_1 \ge 25.0\text{ m}$ and $p_2 \ge 25.0\text{ m}$, and once inside,
  remain in that region continuously for at least $8$ seconds.''} Using
STL~\cite{maler2004monitoring}, this requirement is formally expressed
as the nested formula featuring a spatial conjunction:
\begin{equation}
  \psi = F_{[0, 2.0]} \, \left( G_{[0, 8.0]} \, (p_1 \ge 25.0 \wedge p_2 \ge 25.0) \right)
  \label{eq:motivating_spec}
\end{equation}
where $F_{[a,b]}$ and $G_{[a,b]}$ denote the continuous temporal
\emph{Eventually} (liveness) and \emph{Always} (safety) operators,
respectively.

With initial velocity $v_{\mathrm{init}}=1.0\text{ m/s}$ on each axis,
the greatest admissible displacement is obtained by accelerating to the
speed limit at $t_c=(v_{\max}-v_{\mathrm{init}})/a_{\max}=1.0\text{ s}$
and then cruising:
\begin{equation}
  D(T)=\begin{cases}
    v_{\mathrm{init}}T+\tfrac12 a_{\max}T^2, & 0\le T\le t_c,\\
    v_{\max}T-\dfrac{(v_{\max}-v_{\mathrm{init}})^2}{2a_{\max}}, & T\ge t_c,
  \end{cases}
  \label{eq:motivating_displacement}
\end{equation}
Thus $D(2)=3.5\text{ m}$, and the furthest reachable position at the
deadline is $[23.5,21.5]^T\text{ m}$. The predicate deficits are
$1.5\text{ m}$ and $3.5\text{ m}$, so the original specification is
infeasible. The second axis cannot reach its target before
$3.75\text{ s}$, giving a minimum physical deadline extension of
$1.75\text{ s}$.


\subsubsection{State-of-the-Art MILP Decision Procedure and Temporal Repair}
\label{sec:state-art-milp}

We apply the diagnosis and temporal repair procedure of Ghosh et
al.~\cite{ghosh2016diagnosis}, Section~5.2 and Equation~(16). Its
guarantees concern the discretized synthesis problem; sampled
verification alone does not guarantee continuous-time STL satisfaction
by the physical trajectory. Our Python reimplementation, with sampling
time $\Delta t=0.2\text{ s}$, produces:

\begin{verbatim}
Feasible:  False
Witness:   Maximum optimal slack = 3.50 m
Fix:       F_[3.8,10] G_[3.8,10] P
Time:      1.065109 seconds
Grid:      Verified
Continuous: Failed; robustness = -0.005 m
\end{verbatim}

Here $P\equiv(p_1\ge25.0\wedge p_2\ge25.0)$, and \texttt{Feasible} refers to the
original formula. Time covers the complete diagnosis-and-repair
pipeline.

\paragraph{Dwell and Continuous Satisfaction.}
The repair requires $P$ on $[t+3.8,t+10]$ for some $t\in[3.8,10]$. It
shortens the required dwell from $8\text{ s}$ to $6.2\text{ s}$ and
delays its start relative to the outer witness. Moreover, sampled
validation accepts a controller that violates the repaired formula
continuously. For example, at $t=13.8\text{ s}$ the returned trajectory
has $p_2\approx25\text{ m}$, $v_2\approx-0.1\text{ m/s}$, and
$u_2=1\text{ m/s}^2$ over the next $0.2\text{ s}$. Its segment is
$p_2(13.8+\tau)\approx25-0.1\tau+\tfrac12\tau^2$, for
$0\le\tau\le0.2$, so both endpoint samples are approximately
$25\text{ m}$, but $p_2(13.9)\approx24.995\text{ m}$. An audit
of all continuous outer witness times, including off-grid choices, gives
complete-formula robustness $-0.005\text{ m}$. Hence the returned
controller does \textbf{not} satisfy the repaired STL formula under
continuous-time semantics, despite passing the sampled check.

\subsubsection{Overview of the Proposed Geometric Decision Procedure}
\label{sec:overv-prop-geom}

Our proposed technique (Algorithm~\ref{alg:high_level_decision})
compiles the specification into continuous geometric fields and
evaluates a local matrix-vector inclusion check at the initial state,
without constructing a time grid. Compilation follows the formula
syntax; evaluation of the compiled check does not introduce variables
for individual temporal samples. A violated certificate gives a spatial
deficit that is mapped into an analytical temporal candidate.
Algorithm~\ref{alg:high_level_decision} invokes
Algorithm~\ref{alg:temporal_repair} when this certificate fails. The
first candidate is checked immediately; iteration is the fallback.


\begin{algorithm}[htbp]
  \caption{High-Level Geometric Decision Procedure}
  \label{alg:high_level_decision}
  \footnotesize
  \begin{algorithmic}[1]
    \Require Initial state $x(0)$, specification $\psi$, dynamics, permitted edits and repair budget
    \State \textbf{Geometric Compilation:} Pre-process $\psi$ and compile its complete geometric certificate at $x(0)$.
    \State \textbf{Feasibility Evaluation:} Perform a geometric inclusion check to verify if $x(0)$ is contained within the effective spatial boundaries.
    \If{the complete certificate succeeds}
      \State \Return \texttt{feasibility certified}
    \Else
      \State \textbf{Witness Generation:} Evaluate the maximum positive residual of the compiled requirements.
      \State \textbf{Temporal Repair:} Call Algorithm~\ref{alg:temporal_repair} with the permitted edits and repair budget.
      \State \Return \texttt{original not certified}, Spatial Gap, repair result
    \EndIf
  \end{algorithmic}
\end{algorithm}

\begin{algorithm}[htbp]
  \caption{Analytical-First Iterative Temporal Repair}
  \label{alg:temporal_repair}
  \footnotesize
  \begin{algorithmic}[1]
    \Require Failed specification $\psi$, initial state, permitted interval bounds, attempt budget $K$
    \State $\psi_{\mathrm{cur}}\gets\psi$
    \For{$k=1,\ldots,K$}
      \State Evaluate the complete deficit and branch priorities.\label{alg:repair_evaluate}
      \State Select the largest-weight eligible $F$ branch; permit $G$ only after the maximum permitted $F$-only certificate test fails.\label{alg:repair_select}
      \If{no permitted occurrence can be updated}
        \State \textbf{break}
      \EndIf
      \State Invert the maximum deficit for a candidate; clip the $F$ extension or $G$ contraction to edit bounds.\label{alg:repair_candidate}
      \State Recompile the complete candidate specification $\psi_{\mathrm{new}}$.\label{alg:repair_compile}
      \If{its complete certificate succeeds}
        \State \Return \texttt{repair certified}, $\psi_{\mathrm{new}}$\label{alg:repair_accept}
      \EndIf
      \State Retain monotone effective updates in $\psi_{\mathrm{cur}}$; skip ineffective occurrences.\label{alg:repair_retain}
    \EndFor
    \State \Return \texttt{repair not found}\label{alg:repair_failure}
  \end{algorithmic}
\end{algorithm}

The weights and formula tree locate the repair; the deficit sizes it.
Only upper endpoints change. Safety contraction is permitted after
dwell-preserving certificate search fails within the allowed bounds.

\paragraph{Application of Algorithm~\ref{alg:high_level_decision} to the
  Motivating Example}
The four operations in the algorithm specialize to this task as follows.
\begin{enumerate}
\item \textbf{Step 1 (Geometric Compilation):} Translate the arrival and
  dwell requirements into one conservative spatial certificate using the
  finite-time convergence law developed in
  Section~\ref{sec:bhat-bernstein}. The certificate includes both the
  arrival allowance and the safety buffer for the eight-second dwell.

\item \textbf{Step 2 (Feasibility Evaluation):}
  Evaluating the complete certificate at $x(0)=[20,18,1,1]^T$ gives
  $-11.425188\text{ m}$, so the initial state is not certified.

\item \textbf{Step 3 (Witness Generation):} Report the positive
  certificate deficit, $11.425188\text{ m}$. The geometric priorities
  identify the second position requirement as causing the deficit. This
  deficit includes the conservative dwell buffer and differs from the
  physical arrival shortfall. 

\item \textbf{Step 4 (Temporal Repair):}
  Algorithm~\ref{alg:temporal_repair} selects the enclosing $F$ and maps
  the deficit through the convergence law
  (Section~\ref{sec:bhat-bernstein}) to propose a $14.623436\text{ s}$
  extension. Only the arrival deadline changes, giving
  $\psi'_{\mathrm{geom}}=F_{[0,16.623436]}G_{[0,8]}P$ at the displayed
  precision. Recompilation gives a positive complete certificate of
  $1.035294\text{ m}$, so the first candidate succeeds
  (Algorithm~\ref{alg:temporal_repair}). An independent audit verifies
  that the STL formula is satisfied continuously, unlike the returned
  MILP witness.
\end{enumerate}

The algorithm returns the failed original feasibility check, spatial
deficit, and sufficient deadline extension:

\begin{verbatim}
Feasible:  False
Witness:   Certificate deficit = 11.425188 m
Fix:       Increase F deadline by +14.623436 s
Time:      0.000227960 seconds
Continuous: Verified; dwell = 8.00 seconds
\end{verbatim}

The extension is conservative rather than minimum. The MILP pipeline
takes about $4.67\times10^3$ times longer to compute the temporal fix
compared to our procedure. The example illustrates dwell preservation,
continuous satisfaction, and computation time, alongside the
conservatism of the geometric repair.

After the preliminaries,
Sections~\ref{sec:pre-processing},~\ref{sec:compilation}
and~\ref{sec:polyhedral-bridge} develop geometric compilation and
evaluation. Section~\ref{sec:theor-guar} establishes soundness,
incompleteness, and complexity; Section~\ref{sec:experiments} evaluates
these results.

\section{Preliminaries}
\label{sec:preliminaries}

In this section we give the preliminaries needed to read the rest of the
paper.

\subsection{System Model and STL Syntax}
\label{sec:stl-syntax-semantics}

Consider a continuous-time linear dynamical system governed by:
\begin{equation*}
    \dot{x}(\tau) = \mathcal{F} x(\tau) + \mathcal{G} u(\tau), \quad x(\tau) \in \mathbb{R}^d, \; u(\tau) \in \mathcal{U} \subset \mathbb{R}^m
\end{equation*}
where $x(\tau)$ denotes the physical state trajectory at time instant
$\tau \ge 0$, $u(\tau)$ is the bounded control input, and
$\mathcal{F}, \mathcal{G}$ are constant system matrices. We evaluate Signal Temporal Logic
(STL) specifications built over affine state
predicates~\cite{maler2004monitoring, donze2010robust}.

\paragraph{Syntax.} 
The syntax of the supported input STL grammar is defined inductively as:
\begin{equation*}
  \begin{aligned}
    \psi &::= \top \mid \bot \mid \mu(x) \ge 0 \mid \neg\psi \mid \psi_1\wedge\psi_2 \mid \psi_1\vee\psi_2 \mid G_{[a,b]}\psi \mid F_{[a,b]}\psi.
  \end{aligned}
\end{equation*}
where $\top$ and $\bot$ represent logical truth and falsehood,
$[a,b]$ is a temporal interval with
bounds $0 \le a \le b < \infty$, and $\mu: \mathbb{R}^d \to \mathbb{R}$ is an affine spatial predicate
defining an atomic constraint: $ \mu(x) \triangleq g^T x - h \ge 0$, with normal
vector $g \in \mathbb{R}^d$ and scalar threshold $h \in \mathbb{R}$. The temporal operators
$G_{[a,b]}$ and $F_{[a,b]}$ denote the \emph{Always} (globally) and
\emph{Eventually} (future) operators, respectively.

\paragraph{Robust Semantics}
\label{sec:robust-semantics}
Rather than evaluating formulas to boolean truth values, we adopt the
quantitative (robust) semantics of STL \cite{donze2010robust}. The
robustness degree function $\rho(\psi, x, \tau) \in \mathbb{R}\cup\{-\infty,+\infty\}$ quantifies how strongly a
trajectory $x$ satisfies or violates a specification $\psi$ at time
$\tau \ge 0$. Positive robustness guarantees Boolean satisfaction,
and negative robustness guarantees violation; zero robustness is a
boundary case. Our feasibility certificates guarantee $\rho\ge0$.
In particular, Boolean negation of $\mu(x)\ge0$ requires $\mu(x)<0$,
whereas its nonnegative robustness test includes $\mu(x)=0$.
The robustness function is defined recursively as:
{\small
  \begin{gather*}
    \begin{alignedat}{2}
      \rho(\top, x, \tau) &\triangleq +\infty,
      &\qquad \rho(\mu(x) \ge 0, x, \tau) &\triangleq \mu(x(\tau)) \\[-2pt]
      \rho(\bot, x, \tau) &\triangleq -\infty,
      &\qquad \rho(\neg \psi, x, \tau) &\triangleq -\rho(\psi, x, \tau)
    \end{alignedat} \\[-2pt]
    \begin{aligned}
      \rho(\psi_1 \wedge \psi_2, x, \tau) &\triangleq \min \Big( \rho(\psi_1, x, \tau), \, \rho(\psi_2, x, \tau) \Big) \\[-2pt]
      \rho(\psi_1 \vee \psi_2, x, \tau) &\triangleq \max \Big( \rho(\psi_1, x, \tau), \, \rho(\psi_2, x, \tau) \Big) \\[-2pt]
      \rho(G_{[a,b]} \psi, x, \tau) &\triangleq \min_{\tau' \in [\tau+a, \tau+b]} \rho(\psi, x, \tau') \\[-2pt]
      \rho(F_{[a,b]} \psi, x, \tau) &\triangleq \max_{\tau' \in [\tau+a, \tau+b]} \rho(\psi, x, \tau')
    \end{aligned}
  \end{gather*}
}
When checking the overall synthesis feasibility of a mission beginning
at the initial state $x(0)$, we evaluate robustness at $\tau = 0$, denoted
as $\rho(\psi, x, 0)$.

\subsection{Bhat--Bernstein Dynamics and Level-Set Inversion}
\label{sec:bhat-bernstein}

To construct exact, continuous spatial representations for temporal
logic operators, we must bridge the gap between continuous system
dynamics and strict temporal deadlines. We leverage finite-time
convergence dynamics governed by the fractional power law introduced by
Bhat and Bernstein~\cite{bhat2000finite}.

We formulate a generalized scalar propagation rule to model the
evolution of our continuous spatial predicates. Recalling the
continuously differentiable scalar predicate
$\mu: \mathbb{R}^d \to \mathbb{R}$ from Section~\ref{sec:stl-syntax-semantics}, we bound the
\textit{worst-case} evolution of $\mu(x(t))$ along a closed-loop physical
state trajectory using the scalar fractional differential equation:
\begin{equation}
  \dot{\mu}(t) = -c \cdot \operatorname{sgn}(\mu(t)) |\mu(t)|^{\beta}, \quad c > 0, \; \beta \in (0, 1),
  \label{eq:bhat_differential}
\end{equation}
where $c > 0$ is the positive coefficient of the convergence law, and
$\beta \in (0, 1)$ is the fractional power exponent enforcing the
non-Lipschitz behavior at the boundary $\mu = 0$ required for finite-time
reachability.

Integrating Eq.~\eqref{eq:bhat_differential} directly yields the exact
closed-form settling time $\tau_{s}(\mu_0)$ required for a trajectory
starting at an initial spatial distance $\mu(0) = \mu_0 > 0$ to reach the
target boundary $\mu(\tau) = 0$:
\begin{equation}
  \tau_{s}(\mu_0) = \int_{0}^{\mu_0} \frac{d s}{c \, s^{\beta}} = \frac{\mu_0^{1-\beta}}{c(1-\beta)}
  \label{eq:settling_time_integral}
\end{equation}
We invert this scalar evolution to express a temporal requirement as
an initial level-set threshold.

\begin{lemma}[Level-Set Inversion of Fractional Dynamics]
\label{lem:level_set_inversion}
Let $c>0$, $\beta\in(0,1)$, and $\mu_{\text{target}}\ge0$.
For the scalar dynamics in Equation~\eqref{eq:bhat_differential},
the unique initial level $\mu_0\ge\mu_{\text{target}}$ whose first
reaching time of $\mu_{\text{target}}$ is a prescribed $T\ge0$ is
\begin{equation}
  \mu_0(T, \mu_{\text{target}}) = \left( \mu_{\text{target}}^{1-\beta} + c(1-\beta) T \right)^{\frac{1}{1-\beta}}
  \label{eq:spatial_inversion}
\end{equation}
For $\mu_{\text{target}}=0$, this is the inverse of the settling-time
expression in Equation~\eqref{eq:settling_time_integral}.
\end{lemma}
\begin{proof}
  Until the target is reached, the nonnegative branch satisfies
  $\dot\mu=-c\mu^\beta$. Separating variables gives the first reaching time
  $T=\int_{\mu_{\text{target}}}^{\mu_0}\frac{ds}{c\,s^\beta}
  =\frac{\mu_0^{1-\beta}-\mu_{\text{target}}^{1-\beta}} {c(1-\beta)}$. Rearranging
  yields Equation~\eqref{eq:spatial_inversion}; uniqueness follows
  because the reaching time is strictly increasing in $\mu_0$. The
  integral is finite at a zero target since $\beta<1$, so the same
  calculation includes the settling-time case. For $T=0$, the initial
  level equals the target.
\end{proof}

\section{Pre-processing}
\label{sec:pre-processing}

Before geometric compilation, we transform the input specification into
negation normal form (NNF), where negations occur only on atomic
predicates. Double negation is eliminated, De Morgan's laws exchange
conjunction and disjunction, and temporal duality gives
$\neg F_{[a,b]}\phi\equiv G_{[a,b]}\neg\phi$ and
$\neg G_{[a,b]}\phi\equiv F_{[a,b]}\neg\phi$, retaining the same
intervals~\cite{maler2004monitoring,donze2010robust}.
These rewrites preserve standard quantitative robustness:
$\rho(\operatorname{NNF}(\psi),x,\tau)=\rho(\psi,x,\tau)$.
We also simplify truth and falsehood constants, including
$\neg\top=\bot$, $\neg\bot=\top$, the Boolean constant identities,
and $O_{[a,b]}\top=\top$, $O_{[a,b]}\bot=\bot$ for $O\in\{F,G\}$.
A specification reduced to $\top$ is accepted directly; one reduced
to $\bot$ is rejected directly. 

A negated atomic literal retains its syntax and is assigned the exact
robustness function $-\mu$. For $\mu(x)=g^Tx-h$, this function is
$(-g)^Tx-(-h)$ and remains affine. We never negate a compiled composite
surrogate. In the following sections, specifications and their
subformulas refer to the pre-processed NNF, and
$\llbracket\psi\rrbracket$ denotes compilation after this pre-processing.

\section{Geometric Semantics}
\label{sec:compilation}

Geometric semantics recursively map a preprocessed STL formula $\psi$
to a spatial field
$\llbracket\psi\rrbracket:\mathbb R^d\to\mathbb R\cup\{\pm\infty\}$.
The certified initial-state set is
$\mathcal X_{\mathrm{cert}}\triangleq
\{x_0\in\mathbb R^d\mid\llbracket\psi\rrbracket(x_0)\ge0\}$;
soundness under the controller assumption is established below.
Local compilation requires a finite value and an existing finite
gradient at the query, rather than global differentiability.

\subsection{Geometric Denotational Semantics}
\label{sec:denotational-semantics}

The valuation map operates inductively on the preprocessed syntax tree.
Equation~\eqref{eq:denotational_semantics} defines the compositional
translation rules in point-free functional form:
{\small
  \begin{gather}
    \begin{alignedat}{2}
      \llbracket \top \rrbracket &\triangleq +\infty,
      &\qquad \llbracket \mu(x) \ge 0 \rrbracket &\triangleq \mu \\[-2pt]
      \llbracket \bot \rrbracket &\triangleq -\infty,
      &\qquad \llbracket \neg(\mu(x) \ge 0) \rrbracket &\triangleq -\mu
    \end{alignedat} \nonumber \\[-2pt]
    \begin{aligned}
      \llbracket \psi_1 \wedge \psi_2 \rrbracket
        &\triangleq -\frac{1}{\eta}\ln\!\left(e^{-\eta \llbracket \psi_1 \rrbracket}+e^{-\eta \llbracket \psi_2 \rrbracket}\right) \\[-2pt]
      \llbracket \psi_1 \vee \psi_2 \rrbracket
        &\triangleq \frac{1}{\eta}\ln\!\left(e^{\eta \llbracket \psi_1 \rrbracket}+e^{\eta \llbracket \psi_2 \rrbracket}\right)-\frac{\ln2}{\eta} \\[-2pt]
      \llbracket F_{[a,b]} \psi \rrbracket
        &\triangleq \operatorname{InvertLiveness}(\llbracket \psi \rrbracket,a,b) \\[-2pt]
      \llbracket G_{[a,b]} \psi \rrbracket
        &\triangleq \operatorname{InvertSafety}(\llbracket \psi \rrbracket,a,b)
    \end{aligned}
    \label{eq:denotational_semantics}
  \end{gather}
}
\subsection{Operator Mappings and Level-Set Inversion}
\label{sec:oper-mapp-level}

The core of the geometric semantics relies on mapping temporal operators
$F_{[a,b]}$ and $G_{[a,b]}$ into higher-order functions (functionals)
that take an un-evaluated spatial function
$f: \mathbb{R}^d \to \mathbb{R}$ and return a transformed spatial function over the initial
state $x(0)$. To do this, we leverage the dynamics established in
Section~\ref{sec:bhat-bernstein}. We first define the necessary
geometric set at the start of the active interval (time $a$) to satisfy
the property for the duration $b-a$, utilizing the closed-form spatial
inversion from Equation~\eqref{eq:spatial_inversion}. We then apply a
backward pre-image operator to map this spatial requirement back to time
zero.

\subsubsection{Liveness via Level-Set Inversion (InvertLiveness)}
We define $\text{InvertLiveness}$ as a higher-order operator that
transforms an inner spatial level-set function
$f \triangleq \llbracket \psi \rrbracket$ into a time-zero liveness spatial field.
The target set $\{x \in \mathbb{R}^d \mid f(x) \ge 0\}$ effectively acts as a capture
basin guiding the system.

To guarantee the state reaches the target within the active temporal
window of duration $b-a$, we apply the Bhat--Bernstein drift bound. By
evaluating Equation~\eqref{eq:spatial_inversion} with $T = b - a$ and a
zero-boundary target, the state at the start of the interval, $x(a)$,
must reside within a geometric capture basin parameterised by $f$:
$ \mathcal{S}_{b-a}(f) \triangleq \left\{ x \in \mathbb{R}^d \;\middle\vert\; -f(x) \le
  \big(c(1-\beta)(b-a)\big)^{\frac{1}{1-\beta}} \right\} $. To enforce this
requirement from the initial state $x(0)$, the system must be driven
such that it enters $\mathcal{S}_{b-a}(f)$ at exactly local time
$\tau=a$. We define the $a$-step controlled backward pre-image operator
$Pre_a$, which yields the exact set of initial states at $\tau=0$ from
which there exists an admissible control to reach the capture basin at
time $a$:
\begin{equation*}
    Pre_a(\mathcal{S}_{b-a}(f)) \triangleq \{ x \in \mathbb{R}^d \mid \exists u(\cdot) \in \mathcal{U}, \text{ s.t. } \Phi_a(x, u(\cdot)) \in \mathcal{S}_{b-a}(f) \},
\end{equation*}
where $\Phi_a$ represents the continuous state transition flow. The functional $\text{InvertLiveness}(f, a, b): \mathbb{R}^d \to \mathbb{R}$ is formally defined as the signed distance function $l$ to this pre-image boundary:
\begin{equation}
  \text{InvertLiveness}(f, a, b) \triangleq l_{Pre_a(\mathcal{S}_{b-a}(f))}
  \label{eq:2}
\end{equation}

\subsubsection{Safety via Level-Set Inversion (InvertSafety)}
Similarly, we define $\text{InvertSafety}$ as a higher-order operator
mapping an inner spatial level-set function
$f \triangleq \llbracket \psi \rrbracket$ into a time-zero safety spatial field,
acting as a guard set bounding the system.

To endure the worst-case adversarial drift and guarantee safety
throughout the entire active duration $b-a$, we apply the spatial
inversion from Equation~\eqref{eq:spatial_inversion}. The state at the
beginning of the interval, $x(a)$, must be strictly bounded from below
by a contracted geometric safety buffer:
$ \mathcal{B}_{b-a}(f) \triangleq \left\{ x \in \mathbb{R}^d \;\middle\vert\; f(x) \ge
  \big(c(1-\beta)(b-a)\big)^{\frac{1}{1-\beta}} \right\} $. To satisfy this
delayed constraint from initialization, the trajectory is unconstrained
during $[0, a)$, but it must enter $\mathcal{B}_{b-a}(f)$ at exactly time
$\tau=a$. We apply the backward pre-image operator $Pre_a$ to propagate
this requirement back to time $0$:
\begin{equation*}
  Pre_a(\mathcal{B}_{b-a}(f)) \triangleq \{ x \in \mathbb{R}^d \mid \exists u(\cdot) \in \mathcal{U}, \text{ s.t. } \Phi_a(x, u(\cdot)) \in \mathcal{B}_{b-a}(f) \}.
\end{equation*}
The functional
$\text{InvertSafety}(f, a, b): \mathbb{R}^d \to \mathbb{R}$ is then defined as the signed
distance function to this time-zero pre-image boundary:
\begin{equation}
  \text{InvertSafety}(f, a, b) \triangleq l_{Pre_a(\mathcal{B}_{b-a}(f))}
  \label{eq:1}
\end{equation}

\section{Polyhedral Compilation}
\label{sec:polyhedral-bridge}

We now formalize the mapping from the recursive map evaluated at the
initial state, $\llbracket \psi \rrbracket(x(0)) \ge 0$, to the polyhedral
system $A x(0) \le b_{\text{eff}}$.

\begin{theorem}[Polyhedral Compilation and Local Evaluation]
\label{thm:polyhedral_compilation}
Consider a nonconstant pre-processed STL specification $\psi$ constructed
over affine atomic spatial
predicates, 
under the continuous linear dynamics
defined in Section~\ref{sec:stl-syntax-semantics}. Evaluating the
spatial gradients of the continuous geometric semantics
$\llbracket \psi \rrbracket$ at the nominal operating point $x^* = x(0)$
compiles the formula into a local polyhedral system
$A x(0) \le b_{\text{eff}}$, where $A \in \mathbb{R}^{1 \times d}$ and
$b_{\text{eff}} \in \mathbb{R}^1$. The single row certifies the complete
composed field, whose value and spatial gradient are assumed finite at
the query point.

At the query point $x^* = x(0)$, evaluating this compiled polyhedral
check is algebraically equivalent to evaluating the continuous geometric
semantics:
$ A x(0) \le b_{\text{eff}} \iff \llbracket \psi \rrbracket(x(0)) \ge 0 $, where
$A$ and $b_{\text{eff}}$ are constructed inductively from the system
flow matrices, Bhat--Bernstein level-set inversions, and local gradient
projections of the Log-Sum-Exp compositional operators in
Equation~\eqref{eq:denotational_semantics}.
\end{theorem}

\begin{proof}
First compute the complete geometric field recursively using the affine
literals, temporal inversions, and LSE Boolean operators in
Equation~\eqref{eq:denotational_semantics}. Boolean and temporal
composition are therefore evaluated before forming the linear row.
The first-order Taylor expansion about $x^*$ has affine part
$\llbracket\psi\rrbracket(x^*)+
\nabla\llbracket\psi\rrbracket(x^*)^T(z-x^*)$ for a nearby initial
state $z$. This expression retains the field's value at $x^*$ and uses
its gradient to describe how the certificate changes locally.

Requiring this affine expression to be nonnegative gives a linear
inequality. Where the value and gradient are finite, define
$A=-\nabla\llbracket\psi\rrbracket(x^*)^T$ and
$b_{\text{eff}}=\llbracket\psi\rrbracket(x^*)
-\nabla\llbracket\psi\rrbracket(x^*)^Tx^*$. With these coefficients,
$A x(0)\le b_{\text{eff}}$ is equivalent to
$\llbracket\psi\rrbracket(x^*)+ \nabla\llbracket\psi\rrbracket(x^*)^T(x(0)-x^*)\ge0$.
At the compilation query $x^*=x(0)$, the displacement vanishes, so the
Taylor model reproduces the actual field value exactly. Hence the row
inequality is precisely the nonnegativity check for the complete
composed field at that point.
\end{proof}

With explicit temporal dependencies translated into time-zero spatial
bounds $b_{\text{eff}}$, the control synthesis feasibility check is
decoupled from discrete clock variables.

\subsection{Application to the Motivating Example}
\label{sec:compilation_motivating_example}

We apply the local compilation to the double-integrator task in
Equation~\eqref{eq:motivating_spec}, using the full state
$x=[p_1,p_2,v_1,v_2]^T$ and query point
$x^*=x(0)=[20.0,18.0,1.0,1.0]^T$, with convergence parameters
$c=0.6$, $\beta=0.2$, and smoothing parameter $\eta=10$.

\paragraph{Steps 1 \& 2: LSE Formulation and Temporal Inversion.}
First compile the target conjunction as
$f(p)=-\eta^{-1}\ln\bigl(e^{-\eta(p_1-25)}+e^{-\eta(p_2-25)}\bigr)$. 
Both temporal lower endpoints are zero, so $Pre_0$ is the identity.
Applying the safety contraction and then the liveness expansion gives
$ \llbracket\psi\rrbracket(p,v) =f(p)-\mu_0(8,0)+\mu_0(2,0)$.
Equation~\eqref{eq:spatial_inversion} gives
$\mu_0(8,0)=5.375440\text{ m}$ and $\mu_0(2,0)=0.950253\text{ m}$. These
offsets include the conservative dwell buffer and arrival allowance.

\paragraph{Step 3: Linearization and Matrix Assembly.}
The Log-Sum-Exponential (LSE) weights are
$w_i=e^{-\eta(p_i^*-25)}/\sum_{j=1}^2 e^{-\eta(p_j^*-25)}$
(Equation~\eqref{eq:denotational_semantics}). At the query point,
$w_1\approx2.061\times10^{-9}$ and $w_2=1-w_1$. The temporal offsets are constant
in the state, so $\nabla\llbracket\psi\rrbracket(x^*)=[w_1,w_2,0,0]^T$. The
first-order Taylor row in Theorem~\ref{thm:polyhedral_compilation} is
therefore $A=-\nabla\llbracket\psi\rrbracket(x^*)^T$, with
$b_{\text{eff}}=\llbracket\psi\rrbracket(x^*)
-\nabla\llbracket\psi\rrbracket(x^*)^Tx^* \approx-29.425188\text{ m}$. At
$x^*=x(0)$, the Taylor displacement vanishes, so the row checks the
complete field exactly at this point.

\paragraph{Step 4: Witness Generation and Temporal Fix.}
The positive row residual is
$r=Ax(0)-b_{\text{eff}}=-\llbracket\psi\rrbracket(x^*)
\approx11.425188\text{ m}$. 
The dominant weight $w_2$ identifies the second target for repair, and
Algorithm~\ref{alg:temporal_repair} selects its enclosing $F$ for
extension. Its first candidate uses
Equation~\eqref{eq:settling_time_integral}:
$ \Delta T^*=\tau_s(r)=\frac{r^{1-\beta}}{c(1-\beta)} \approx14.623436\text{ s}$. This gives
$\psi'=F_{[0,16.623436]}G_{[0,8]}P$ at the displayed precision. Recompiling
the complete candidate yields
$ \llbracket\psi'\rrbracket(x^*) =-r+\mu_0(2+\Delta T^*,0)-\mu_0(2,0)
\approx1.035294\text{ m}>0$. 

\section{Analysis of the Decision Procedure}
\label{sec:theor-guar}

In this section we establish soundness, quantified incompleteness, and
derive complexity bounds for the feasibility decision procedure.

\subsection{Formal Soundness Guarantees}
\label{sec:soundness}

In this section, we establish the formal soundness of the proposed
geometric decision procedure. 


\begin{lemma}[Nested-Level Temporal Horizon Additivity]
\label{lem:nested_composition}
Let $\Pi$ be the root-to-leaf paths in the syntax tree of $\psi$,
where each temporal node $O_i$ has interval $[a_i,b_i]$.
The syntactic evaluation horizon is
$T_{\text{eff}}(\psi)\triangleq
\max_{\pi\in\Pi}\sum_{O_i\in\pi}b_i$.
Constants, atomic predicates, and Boolean operators contribute zero.
A trajectory defined on $[\tau,\tau+T_{\text{eff}}(\psi)]$
provides sufficient coverage to evaluate $\psi$ at $\tau$.
This horizon uses upper endpoints, while the individual temporal
intervals remain in force.
\end{lemma}

\begin{proof}
  The standard horizon construction adds $b$ for each temporal operator,
  takes the maximum under conjunction or disjunction, and leaves the
  horizon unchanged under negation~\cite[Sec.~2]{maler2004monitoring}.
  Expanding this construction yields the stated maximum path sum.
\end{proof}

\begin{assumption}[Joint Dynamic Realizability]
\label{assum:low_level_tracking}
For the preprocessed specification $\psi$ and fixed $\eta>0$,
assume every $x_0$ satisfying $\llbracket\psi\rrbracket(x_0)\ge0$
admits one controller acting on the complete compiled surrogate,
generating an admissible $u_\psi$ and an absolutely continuous
$x_\psi$ on $[0,T_{\text{eff}}(\psi)]$, with $x_\psi(0)=x_0$.
For each temporal obligation $\theta=O_{[a,b]}\phi$ with
$\llbracket\theta\rrbracket(x_\psi(\tau))\ge0$ and
$0\le\tau\le T_{\text{eff}}(\psi)-T_{\text{eff}}(\theta)$,
the same control realizes the corresponding controlled-preimage
transition at $\tau+a$. Its finite trace
$s(t)=\llbracket\phi\rrbracket(x_\psi(t))$ is absolutely continuous and
satisfies $\dot s\ge-c\,\operatorname{sgn}(s)|s|^\beta$ almost everywhere
on $[\tau+a,\tau+b]$, using the inversion parameters.
These requirements hold jointly, without slack, for overlapping
obligations and their continuations. 
\end{assumption}

Controller abstraction enables reuse across control architectures; this
assumption defines the contract with the execution controller. Our
experiments construct and continuously validate a realization based on
the higher-order controller of Xiao et al.~\cite{xiao2021highorder}.

\begin{theorem}[Geometric Soundness]
\label{thm:geometric_soundness}
Consider the continuous-time dynamics and supported STL grammar of
Section~\ref{sec:stl-syntax-semantics}, after pre-processing as in
Section~\ref{sec:pre-processing}. For a nonconstant specification,
let $A$ and $b_{\text{eff}}$
be compiled for the complete specification $\psi$ at
$x^*=x_0=x(0)$ as in Theorem~\ref{thm:polyhedral_compilation}.
Under Assumption~\ref{assum:low_level_tracking},
\[
  \begin{aligned}
    A x_0\le b_{\text{eff}}
    &\quad\Longleftrightarrow\quad
      \llbracket\psi\rrbracket(x_0)\ge0\\
    &\quad\Longrightarrow\quad
      \exists\,u(\cdot)\text{ admissible}:\;
      \rho(\psi,x^{x_0,u},0)\ge0.
  \end{aligned}
\]
Here $x^{x_0,u}$ is the trajectory starting from $x_0$ under $u$, and
$\rho$ is evaluated over continuous temporal intervals. Thus a
successful local matrix check certifies an admissible control with
nonnegative continuous robustness over the complete specification
horizon $T_{\text{eff}}(\psi)$.
\end{theorem}

\begin{proof}
Suppose $A x_0\le b_{\text{eff}}$. At the compilation query point,
Theorem~\ref{thm:polyhedral_compilation} gives
$\llbracket\psi\rrbracket(x_0)\ge0$. Fix the single admissible trajectory
$x_\psi$ supplied by Assumption~\ref{assum:low_level_tracking}.
By structural induction, we establish
\[
  \llbracket\phi\rrbracket(x_\psi(\tau))\ge0
  \quad\Longrightarrow\quad
  \rho(\phi,x_\psi,\tau)\ge0
\]
for each subformula $\phi$ and every $\tau\ge0$ with
$\tau+T_{\text{eff}}(\phi)\le T_{\text{eff}}(\psi)$, always on this
same trajectory.

\paragraph{Base Cases.}
Truth holds trivially; falsehood has a negative geometric value, so
the implication holds vacuously. For an atomic predicate, the geometric
field and STL robustness both equal $\mu(x_\psi(\tau))$; for a negated
atomic literal, both equal $-\mu(x_\psi(\tau))$.

\paragraph{Logical Operators.}
The conservative LSE conjunction is bounded above by the minimum of
the geometric values of its two subformulas. Its nonnegativity therefore
implies that both subformula values are nonnegative; their induction
hypotheses give nonnegative
conjunction robustness. Similarly, the normalized LSE disjunction is
bounded above by the maximum of the geometric values of its two
subformulas. At least one subformula value must then be nonnegative,
and its induction hypothesis establishes
nonnegative disjunction robustness.

\paragraph{Temporal Operators.}
Assume the induction hypothesis holds for the inner formula $\phi$;
an identically $+\infty$ inner field is immediate.
For $F_{[a,b]}\phi$, Equation~\eqref{eq:2} and the assumed realization
of the controlled preimage place $x_\psi(\tau+a)$ in
$\mathcal S_{b-a}(\llbracket\phi\rrbracket)$. If the inner field is
already nonnegative there, take $\tau'=\tau+a$. Otherwise, its deficit
is within the capture bound, so the assumed fractional differential
bound and Equation~\eqref{eq:settling_time_integral} guarantee a
settling time no greater than $b-a$. Hence the inner field is
nonnegative at some $\tau'\in[\tau+a,\tau+b]$. Its induction hypothesis
gives $\rho(\phi,x_\psi,\tau')\ge0$, establishing
$\rho(F_{[a,b]}\phi,x_\psi,\tau)\ge0$.

For $G_{[a,b]}\phi$, Equation~\eqref{eq:1} and the same controller
assumption place $x_\psi(\tau+a)$ in
$\mathcal B_{b-a}(\llbracket\phi\rrbracket)$. By the safety-buffer
inversion in Equation~\eqref{eq:spatial_inversion}, the assumed lower
differential bound prevents the inner field from becoming negative
during $[\tau+a,\tau+b]$. The induction hypothesis therefore gives
$\rho(\phi,x_\psi,\tau')\ge0$ at every time in that interval, establishing
$\rho(G_{[a,b]}\phi,x_\psi,\tau)\ge0$. The case $a=b$ reduces to the
nonnegative field value at the single evaluation time.

For either temporal operator, Lemma~\ref{lem:nested_composition} ensures
that every invoked child horizon fits within the available trajectory,
since $\tau'\le\tau+b$ and
$T_{\text{eff}}(O_{[a,b]}\phi)=b+T_{\text{eff}}(\phi)$ for $O\in\{F,G\}$.
Applying the induction result to $\psi$ at $\tau=0$ gives
$\rho(\psi,x_\psi,0)\ge0$. Its admissible control signal witnesses the
claimed continuous-time feasibility.
Since pre-processing preserves standard quantitative robustness, the
same guarantee holds for the original input specification.
\end{proof}

Ghosh et al.~\cite{ghosh2016diagnosis} formulate synthesis and repair
over a discrete-time signal from a sampled system model
(Section~2.1 of their paper), with soundness relative to that synthesis
model. Satisfaction of the sampled constraints alone does not establish
STL satisfaction between samples without additional conditions on the
intersample behavior. Theorem~\ref{thm:geometric_soundness} instead
certifies continuous-time feasibility from the local matrix check under
the stated controller assumption.

\subsection{Temporal Repair and its Formal Guarantees}
\label{sec:repair_guarantees}

\paragraph{Description of Algorithm~\ref{alg:temporal_repair}.}
The algorithm proposes one temporal update at a time. At each attempt,
line~\ref{alg:repair_evaluate} uses the current local compilation at
$x_0=x(0)$ and computes $y=A x_0-b_{\text{eff}}$ and
$r=\max_i\max\{y_i,0\}$. A row attaining $r$ identifies the failed
requirement. For subformula values $z_i$, conjunction uses
$w_i=e^{-\eta z_i}/\sum_j e^{-\eta z_j}$ and disjunction uses
$w_i=e^{\eta z_i}/\sum_j e^{\eta z_j}$
(Equation~\eqref{eq:denotational_semantics}). Products along Boolean
paths give priorities for temporal occurrences.
Line~\ref{alg:repair_select} chooses the largest-weight ($w_{i}$)
eligible $F$, preferring the outermost occurrence on a tied branch, then
syntax order.

Eligibility accounts for smoothing as well as negative values. Start
with target zero for the complete certificate. At a Boolean node with
$N_{\mathrm B}$ operands, add $\ln(N_{\mathrm B})/\eta$ to the operand
targets (Lemma~\ref{lem:lse_conservatism}). In the algebraic
implementation, passing a target inside a temporal node subtracts its
signed offset: $+\mu_0(b-a,0)$ for $F$ and $-\mu_0(b-a,0)$ for $G$. An
occurrence is eligible while its value is below its target and it has a
permitted edit. Thus positive subformula values may still need repair
when their complete LSE composition is negative.

Line~\ref{alg:repair_candidate} proposes
$\Delta T=\tau_s(r)=r^{1-\beta}/(c(1-\beta))$
(Equation~\eqref{eq:settling_time_integral}) and clips it to the edit
bounds. It keeps $a$ fixed, increasing $b$ for $F$ or decreasing it for
$G$ with $b'\ge a$. 
Line~\ref{alg:repair_compile} checks the complete recompiled formula
(Theorem~\ref{thm:polyhedral_compilation}); line~\ref{alg:repair_accept}
accepts only a nonnegative certificate. Otherwise,
line~\ref{alg:repair_retain} retains effective improvements, recomputes
the deficit and weights, and skips ineffective occurrences.

When eligible $F$ attempts are exhausted, line~\ref{alg:repair_select}
tests all $F$ endpoints at their permitted maxima. A successful probe
continues the $F$ search; a failed probe permits weight-prioritized $G$
contraction. The probe selects the search branch; actual updates still
change one occurrence at a time. By
Lemma~\ref{lem:repair_monotonicity}, failure rules out an $F$-only
certificate within these edit bounds, without establishing physical
necessity of contraction.

We use the following examples to show working of
Algorithm~\ref{alg:temporal_repair}. In these examples we use $c=0.6$,
$\beta=0.2$, and $\eta=10$. The additional UAV cases have
$\|u\|_\infty\le3\text{ m/s}^2$ and $\|v\|_\infty\le5\text{ m/s}$, with up to
$3000\text{ s}$ extension per $F$ and minimum $G$ duration
$0.01\text{ s}$.

\begin{enumerate}
\item \textbf{The first candidate succeeds.} In the motivating example
  (Sections~\ref{sec:motivating_example}
  and~\ref{sec:compilation_motivating_example}), the dominant weight
  identifies the second target and its enclosing $F$. The
  $14.623436\text{ s}$ extension gives complete certificate
  $1.035294\text{ m}>0$, so line~\ref{alg:repair_accept} returns
  immediately, preserving $G_{[0,8]}$.

\item \textbf{Another $F$ update is needed.} From $x_0=0$, consider
  $\psi_{\mathrm{two}}=F_{[1,1.4]}G_{[0,2]}(p_1\ge4) \wedge
  F_{[1,1.5]}G_{[0,2]}(p_2\ge4.6)$. The branch values
  (\(\llbracket\psi_i\rrbracket(x(0))\), for $i \in \{1, 2\}$) are
  $-3.323158$ and $-3.882270\text{ m}$, with weights $0.003717$ and
  $0.996283$. The second deadline is therefore extended first, to
  $7.666812\text{ s}$. The complete value improves from $-3.882642$ to
  $-3.323158\text{ m}$ but remains negative.
  Line~\ref{alg:repair_retain} keeps this update; refreshed weights
  select the first branch. Extending its deadline to $6.845021\text{ s}$
  gives complete value $0.133023\text{ m}>0$. Both two-second dwells are
  preserved.

\item \textbf{$F$ extensions cannot repair the remaining failure.} Let
  $x_0=(0.996,0,0,0.17,0,0)^T$ and
  $\psi_{\mathrm{safety}}=G_{[0,0.2]}(p_1\le1) \wedge F_{[1,3]}G_{[0,2]}(p_2\ge6)$,
  where $p_1\le1$ abbreviates $1-p_1\ge0$. The arrival branch is initially
  dominant. Extending its deadline to $9.939509\text{ s}$ leaves
  complete value (\(\llbracket\psi_{\mathrm{safety}}\rrbracket(x(0))\))
  $-0.049480\text{ m}$. Even the maximum-$F$ probe leaves the
  independent wall certificate negative. One $G$ contraction to
  $[0,0.0119383]$ then gives complete value $0.002351\text{ m}>0$, while
  preserving the arrival dwell $G_{[0,2]}$.

\end{enumerate}

$K$ bounds candidate attempts, including unsuccessful candidates;
maximum-$F$ probes are additional certificate evaluations. Exhausting
the attempt budget or permitted occurrences returns \texttt{repair not
  found}. Returning \texttt{repair not found} does not imply physical
infeasibility.

\begin{lemma}[Monotonicity of Temporal Relaxation]
\label{lem:repair_monotonicity}
Fix the dynamics, admissible controls, inversion parameters, smoothing
parameter, and distance convention. For a pre-processed NNF
specification, let $\psi'$ increase an $F$ upper endpoint or decrease
a $G$ upper endpoint, keeping its lower endpoint fixed and its interval
nonempty. Then
$\llbracket\psi'\rrbracket(x)\ge\llbracket\psi\rrbracket(x)$.
\end{lemma}
\begin{proof}
The threshold $\mu_0(T,0)$ increases with $T$. Thus extending $F$
enlarges $\mathcal S_{b-a}(f)$, while shortening $G$ enlarges
$\mathcal B_{b-a}(f)$. The controlled preimage preserves set inclusion.
For $C\subseteq D$, using the same metric and state domain, the
inside-positive signed distance
$l_C(x)=\operatorname{dist}(x,C^c)-\operatorname{dist}(x,C)$ satisfies
$l_C(x)\le l_D(x)$: enlarging the set cannot decrease distance to its
complement or increase distance to the set. If $f\le g$, both
$\mathcal S_T(f)\subseteq\mathcal S_T(g)$ and
$\mathcal B_T(f)\subseteq\mathcal B_T(g)$, so temporal inversions also
preserve the increase through ancestors. Both LSE Boolean operators
are increasing in every argument. Propagating these inequalities
through the NNF syntax tree proves the complete-field inequality.
\end{proof}

\begin{theorem}[Witness Soundness of Iterative Temporal Repair]
\label{thm:witness_soundness}
Suppose Algorithm~\ref{alg:temporal_repair} returns
\texttt{repair certified} with complete specification $\psi'$.
Let $A'$ and $b_{\text{eff}}'$ be its fresh local compilation at
$x^*=x_0=x(0)$. Under Assumption~\ref{assum:low_level_tracking} for
$\psi'$, there exists an admissible $u(\cdot)\in\mathcal U$ such that
$\rho(\psi',x^{x_0,u},0)\ge0$ over continuous temporal intervals and the
complete repaired horizon $T_{\text{eff}}(\psi')$.
Recompilation may change both $A$ and $b_{\text{eff}}$; the certificate
is evaluated using the repaired pair.
\end{theorem}

\begin{proof}
  Acceptance requires the complete check $A'x_0\le b_{\text{eff}}'$. By
  Theorem~\ref{thm:polyhedral_compilation} at the repaired compilation
  query, this is equivalent to $\llbracket\psi'\rrbracket(x_0)\ge0$.
  Under Assumption~\ref{assum:low_level_tracking} for $\psi'$,
  Theorem~\ref{thm:geometric_soundness} supplies one admissible control
  with nonnegative continuous robustness for $\psi'$ over its complete
  repaired horizon.
\end{proof}


\subsection{Incompleteness}
\label{sec:incompleteness}

The procedure provides a sufficient feasibility certificate under
Assumption~\ref{assum:low_level_tracking}, but may reject continuously
feasible states. We quantify the Boolean smoothing contribution
separately from conservatism in the temporal inversions.

\begin{lemma}[LSE Boolean Conservatism]
  \label{lem:lse_conservatism}
  Let $N_{\mathrm{B}}$ be the number of operands in a Boolean composition,
  and let $\psi_1,\ldots,\psi_{N_{\mathrm{B}}}$ be the subformulas combined by a conjunction
  or disjunction, and let $z_i=\llbracket\psi_i\rrbracket(x)$ be their
  finite geometric-semantic values at a fixed state $x$, after temporal
  compilation. For
  $\eta>0$, a conjunction uses
  $\widetilde B=-\eta^{-1}\ln\bigl(\sum_{i=1}^{N_{\mathrm{B}}} e^{-\eta z_i}\bigr)$,
  whereas a normalized disjunction uses
  $\widetilde B=\eta^{-1}\ln\bigl(\sum_{i=1}^{N_{\mathrm{B}}} e^{\eta z_i}\bigr)
  -\ln(N_{\mathrm{B}})/\eta$. Let $B=\min_i z_i$ for conjunction and
  $B=\max_i z_i$ for disjunction. In either case,
  $0\le B-\widetilde B\le\ln(N_{\mathrm{B}})/\eta$.
  Thus $B\ge\ln(N_{\mathrm{B}})/\eta$ guarantees $\widetilde B\ge0$, and any
  rejection attributable solely to this smoothing occurs within
  $0\le B<\ln(N_{\mathrm{B}})/\eta$. The bound applies to one Boolean composition;
  the binary operators in Equation~\eqref{eq:denotational_semantics}
  correspond to $N_{\mathrm{B}}=2$.
\end{lemma}

\begin{proof}
  Apply the standard LSE bound~\cite[Sec.~3.1.5, p.~72]{boyd2004convex}
  to $-\eta z_i$ for conjunction, then negate and divide by $\eta$. This gives
  $B-\ln(N_{\mathrm{B}})/\eta\le\widetilde B\le B$. For disjunction, applying
  the same bound to $\eta z_i$ gives
  $B\le\eta^{-1}\ln\bigl(\sum_i e^{\eta z_i}\bigr) \le B+\ln(N_{\mathrm{B}})/\eta$;
  subtracting $\ln(N_{\mathrm{B}})/\eta$ yields the same interval for
  $\widetilde B$. The claims follow.
\end{proof}

\begin{proposition}[Incompleteness of the Decision Procedure]
\label{prop:incompleteness}
For every finite $\eta>0$, the geometric decision procedure is
not complete in general: some continuously feasible initial states
fail its geometric certificate and corresponding local
polyhedral check.
\end{proposition}

\begin{proof}
  Consider $\dot x=u$, $\|u\|_\infty\le1$, and
  $\psi=(x_1\ge0)\wedge(x_2\ge0)$. Choose $\zeta=\ln(2)/(2\eta)>0$ and
  $x_0=(\zeta,\zeta)^T$. The admissible input $u(t)=0$ gives
  $\rho(\psi,x,0)=\zeta>0$. For these equal subformula values, the upper error
  bound of Lemma~\ref{lem:lse_conservatism} is attained:
  $\llbracket\psi\rrbracket(x_0) =\zeta-\ln(2)/\eta=-\zeta<0$. The geometric
  certificate therefore rejects this feasible state.
  Theorem~\ref{thm:polyhedral_compilation} gives the same rejection for
  the polyhedral check compiled at $x^*=x_0$.
\end{proof}

The bound $\ln(N_{\mathrm{B}})/\eta$ concerns Boolean smoothing with the
compiled subformula values held fixed; it is not a bound on total
incompleteness. 

\subsection{Computational Complexity}
\label{sec:computational_complexity}

We separate evaluation of an already compiled certificate from the
compilation and repeated checks required by temporal repair; $n=1$ for
one complete formula, while $n>1$ denotes a collection of complete
certificate checks. Costs are arithmetic operations under
fixed-precision evaluation.

\begin{theorem}[Horizon-Independent Decision Procedure]
\label{thm:constant_time}
Given the locally compiled representation $A\in\mathbb R^{n\times d}$ and
$b_{\text{eff}}\in\mathbb R^n$ at the query state $x_0=x(0)$: (1)
Computing $y=A x_0-b_{\text{eff}}$ and checking
$A x_0\le b_{\text{eff}}$ require $O(nd)$ operations. (2) Given $y$,
evaluating $r=\max_i\max\{y_i,0\}$ and, when $r>0$, recording a
maximally violated row require $O(n)$ additional operations. The
combined evaluation cost is $O(nd)$, independent of
$T_{\text{eff}}(\psi)$ for fixed compiled dimensions $n,d$. Compilation is
accounted for separately.
\end{theorem}
\begin{proof}
Matrix-vector multiplication costs $O(nd)$; subtraction and comparison
cost $O(n)$. A single scan of the retained residuals finds their maximum
positive part and its row index in $O(n)$. These operations depend on
the compiled dimensions, with no temporal sampling loop.
\end{proof}

\begin{theorem}[Computational Complexity of Temporal Repair]
\label{thm:repair_complexity}
For Algorithm~\ref{alg:temporal_repair}, let $m$ be the syntax-tree
size, $C_0$ the initial compilation cost, and $C_{\mathrm{upd}}$ an
upper bound on preparing and compiling a candidate and any maximum-$F$
probe per attempt. Let $n,d$ bound the compiled dimensions throughout
repair. Then: (1) Given $r$, $c>0$, and $\beta\in(0,1)$, computing
$\Delta T=r^{1-\beta}/(c(1-\beta))$ and clipping it to the permitted endpoint bounds
cost $O(1)$ scalar evaluations. (2) From the evaluated subformula
values, propagating weights and selecting a temporal occurrence cost
$O(m)$ per attempt; each complete compiled inclusion and deficit check
costs $O(nd)$. (3) With an attempt budget $K\ge1$, the total repair cost
is $O(C_0+K(C_{\mathrm{upd}}+nd+m))$.

The scalar candidate and compiled checks introduce no dependence on
temporal discretization; total repair cost also includes compilation
and the number of attempts.
\end{theorem}
\begin{proof}
The candidate uses a fixed number of scalar operations, including one
fractional power; clipping requires a fixed number of comparisons.
Weight propagation and occurrence selection visit the syntax tree in
$O(m)$ operations. Theorem~\ref{thm:constant_time} gives $O(nd)$ for
each compiled check. Charging candidate and probe preparation and
compilation to $C_{\mathrm{upd}}$, and summing over at most $K$
attempts, gives the stated bound together with initial compilation.
\end{proof}


\section{Experimental Evaluation}
\label{sec:experiments}

We evaluate diagnostic repair, computational cost, and conservative
classification on a UAV with continuous 6D double-integrator dynamics,
$\dot p=v$, $\dot v=u$, and $x=[p^T,v^T]^T$. The per-axis bounds are
$\|u\|_\infty\le a_{\max}=3\text{ m/s}^2$ and
$\|v\|_\infty\le v_{\max}=5\text{ m/s}$. The compiler
(Section~\ref{sec:compilation}) uses $c=0.6$, $\beta=0.2$, and $\eta=10$.
The geometric implementation is restricted to position literals
$g_i^Tp-h_i\ge0$ with
$\|g_i\|_1=1$ and $\sigma_jg_{ij}\ge0$ for one common
$\sigma\in\{-1,1\}^3$, using relative inside-positive signed distances on
$\mathcal X=\mathbb R^3\times[-v_{\max},v_{\max}]^3$ in the norm
$\|(\Delta p,\Delta v)\|_\kappa=\|\Delta p\|_\infty+\kappa\|\Delta v\|_\infty$,
where $\kappa>0$ is at least the largest root-to-leaf sum of temporal
lower endpoints.

Decision-procedure benchmarks use Python 3.12.7 and Gurobi 13.0.3 on an
Intel i7-14700 Linux machine with approximately $31$ GiB RAM. The proposed
and MILP~\cite{ghosh2016diagnosis} procedures receive the same affinity to 24
CPUs, per-library or solver thread ceilings of 24, a $24$ GiB
resident-memory guard, and a $120\text{ s}$ case wall budget including
startup and validation. 

Reported algorithm times include the complete pipeline from compilation
to feasibility check, excluding startup and external continuous audits.
Proposed geometric technique measurements use seven batches of 100
complete repairs per repetition; procedure times in tables and plots are
medians across three repetitions. Unfinished runs retain their status
and are excluded from completion-time medians. Geometric acceptance uses
the complete local matrix certificate. Independent trajectory audits
validate continuous satisfaction; Phase~1 additionally tests a
constructive realization of a controller in
Assumption~\ref{assum:low_level_tracking}.

The evaluation is performed in three phases: Phase~1: Diagnostic Witness
and Temporal Repair. Phase~2: Horizon Scalability and Syntactic
Complexity. Phase~3: Soundness and Conservative Rejections.

\subsection{Phase 1: Temporal Fix and Diagnostic Witness Validation}
\label{sec:exp_phase1}

\textbf{Scenario and Specification Setup:} The initial state is \\
$x_0=[20,18,5,0,0,0]^T$, and
$P\equiv(p_1\ge25)\wedge(p_2\ge25)\wedge(p_3\ge20)$ defines a one-sided target region.
Arrival during $[1,3]$ followed immediately by eight seconds of dwell is
specified by the Phase~1 specification $\psi=F_{[1,3]}G_{[0,8]}P$.
For initial velocity $v$, the maximum admissible displacement is
\begin{equation}
  \begin{aligned}
    D(T;v)&=v\tau+\tfrac12 a_{\max}\tau^2
                +v_{\max}(T-\tau),\\
    \tau&=\min\{T,(v_{\max}-v)/a_{\max}\}.
  \end{aligned}
  \label{eq:uav_displacement}
\end{equation}
Thus $D(3;0)=10.8333\text{ m}$, and the vertical target's physical
shortfall is $15-D(3;0)=4.1667\text{ m}$.

\begin{table}[t]
  \caption{Phase~1 repair comparison. MILP uses $\Delta t=0.1$ s. Diagnostics
    are optimal predicate slack for MILP and the compiled certificate
    deficit for geometry. Continuous validity covers the MILP witness
    and both proposed analytical and HOCLBF geometric trajectories;
    procedure times exclude trajectory generation.}
  \label{tab:phase1_results}
  \centering
  \footnotesize
  \begin{tabular}{l c c}
    \toprule
    \textbf{Metric} & \textbf{MILP} & \textbf{Geometric} \\
    \midrule
    Slack / certificate deficit & $4.1700$ m & $17.9252$ m \\
    Repaired $F$ interval & $[3.9,11]$ s & $[1,23.9667]$ s \\
    Repaired $G$ interval & $[3.9,11]$ s & $[0,8]$ s \\
    Dwell duration & $7.1$ s & $8.0$ s \\
    Continuous trajectory & Violates & Both satisfy \\
    Procedure time (ms) & $4451.320$ & $0.264923$ \\
    \bottomrule
  \end{tabular}
\end{table}

\textbf{Diagnostic Witness and Temporal Repair Comparison:}
Let $q=(25,25,20)^T$, $z_j=p_{0,j}+D(1;v_{0,j})-q_j$, and
$S_\eta(z)=-\eta^{-1}\ln(\sum_{j=1}^3e^{-\eta z_j})$.
Applying the fractional operators in
Equation~\eqref{eq:denotational_semantics} gives
\begin{equation}
  \llbracket\psi\rrbracket(x_0) =S_\eta(z)-\mu_0(8,0)+\mu_0(2,0).
  \label{eq:experimental_fractional_field}
\end{equation}
The controlled preimage accounts for the first second; the remaining
liveness window and safety dwell contribute their respective level-set
inversions. At the initial state, the complete field is
$-17.9252\text{ m}$, giving $r=17.9252\text{ m}$.
Algorithm~\ref{alg:temporal_repair} computes
$\Delta T=r^{1-\beta}/(c(1-\beta))=20.9667\text{ s}$ and extends only the outer
deadline. Its first candidate,
$\psi'_{\mathrm{geom}}=F_{[1,23.9667]}G_{[0,8]}P$, has complete certificate
value $1.2120\text{ m}>0$ and passes the fresh local matrix check. 

The MILP pipeline returns
$\psi'_{\mathrm{MILP}}=F_{[3.9,11]}G_{[3.9,11]}P$ at
$\Delta t=0.1\text{ s}$. It changes both intervals and reduces dwell from
$8$ to $7.1$ seconds. Its maximum optimal predicate slack is
$4.17\text{ m}$, while the geometric deficit includes the fractional
safety buffer. Table~\ref{tab:phase1_results} reports a median
$4.4513\text{ s}$ MILP runtime versus $0.264923\text{ ms}$ for the
proposed geometric decision procedure, a ratio of approximately
$1.68\times10^4$. At $\Delta t=0.02\text{ s}$ all three MILP runs leave IIS
diagnosis unresolved; at $0.005\text{ s}$ all three reach the memory
guard.

\begin{figure*}[htbp]
  \centering
  \begin{subfigure}[b]{0.32\textwidth}
    \centering
    \includegraphics[width=\textwidth]{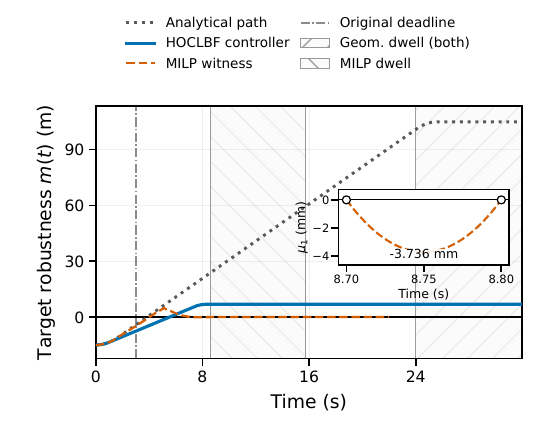}
    \caption{Target robustness: dotted proposed analytical path, solid HOCLBF
      controller, and dashed MILP witness. Hatching marks dwells;
      inset: violation between samples.}
    \label{fig:phase1_trajectory}
  \end{subfigure}
  \hfill
  \begin{subfigure}[b]{0.32\textwidth}
    \centering
    \includegraphics[width=\textwidth]{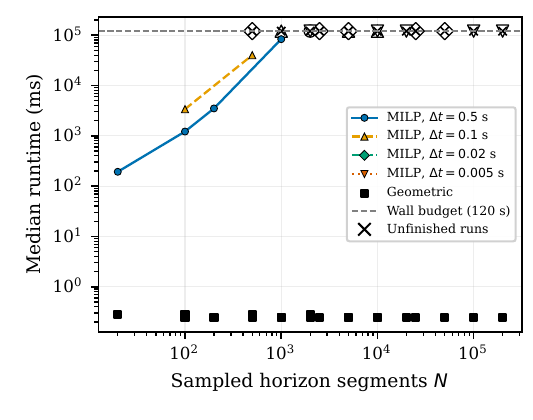}
    \caption{Median complete algorithm times versus original horizon
      segments $N=T/\Delta t$. Crosses indicate unfinished runs,
      placed at the wall-budget reference height.}
    \label{fig:phase2_scalability}
  \end{subfigure}
  \hfill
  \begin{subfigure}[b]{0.32\textwidth}
    \centering
    \includegraphics[width=\textwidth]{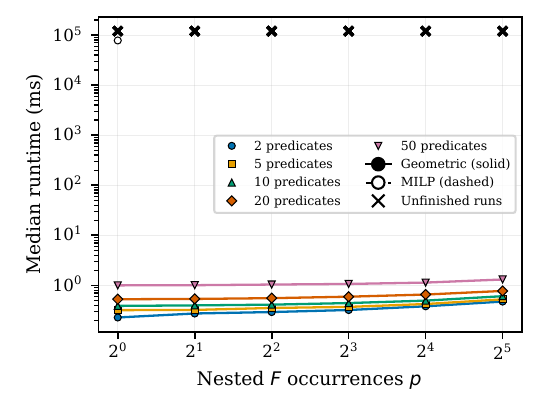}
    \caption{Median complete algorithm times versus nested $F$ count
      $p$. Markers identify predicate count; crosses mark unfinished
      MILP runs, rather than completion times.}
    \label{fig:phase3_nesting}
  \end{subfigure}
  \caption{Continuous trajectories and controller realization, horizon
    scalability, and syntactic complexity on the double-integrator benchmarks.}
  \label{fig:combined_experiments}
\end{figure*}

\textbf{Continuous-Trajectory Validation:}
Figure~\ref{fig:phase1_trajectory} compares the proposed geometric
analytical feasible path, the controller-generated path, and the MILP
witness using instantaneous target robustness
$m(t)=\rho(P,x,t)=\min\{p_1(t)-25,p_2(t)-25,p_3(t)-20\}$, in metres.
Negative values before the selected dwell are permitted by $F$. Opposite
hatching marks the shared geometric dwell $[23.9667,31.9667]$ and the
MILP dwell $[8.6,15.7]$, obtained by adding $[3.9,11]$ to its outer
witness time $4.7$. The proposed geometric analytical path accelerates
to the speed limit, cruises until the repaired deadline, brakes, and
holds; its minimum dwell robustness is $100.6667\text{ m}$.

A separate one-thread run adapts the higher-order control (HOCLBF)
construction of Xiao et al.~\cite{xiao2021highorder} to the complete
dwell field $s(p)=f(p)-\mu_0(8,0)$, where $f(p)=S_\eta(p-q)$, with
$0.01\text{ s}$ updates and unchanged $c,\beta,\eta$. The first second realizes
the controlled preimage to $[21.5,19.5,6.5,3,3,3]^T$. The controller
reaches the dwell buffer at $7.291813\text{ s}$, then brakes and holds.
All $3198$ QPs solve without slack; maximum acceleration and speed are
$3\text{ m/s}^2$ and $3\text{ m/s}$. 
Continuous interval audits verify the fractional reaching and safety
bounds and satisfaction throughout the eight-second dwell. Its minimum
dwell robustness is $6.897332\text{ m}$. 

For the MILP witness, exact constant-acceleration interpolation of
$\mu_1(t)=p_1(t)-25$ gives
$\mu_1(8.7)=\mu_1(8.8)=0$ and $\mu_1(8.75)=-0.003736\text{ m}$. The full
continuous audit considers every outer witness time and finds that none
satisfies the repaired formula. Thus the returned MILP control witness
satisfies the sampled model while violating continuous-time STL. Both
geometric paths satisfy the repaired formula continuously; the HOCLBF
run additionally realizes the controller conditions for this instance,
providing an empirical check of Theorems~\ref{thm:geometric_soundness}
and~\ref{thm:witness_soundness}.

\paragraph{Iterative Repair Validation.}
We validate Algorithm~\ref{alg:temporal_repair} on cases requiring
repeated deadline extensions and safety contraction. A multi-target case
combines several target requirements for one UAV, each with its own
arrival deadline and two-second dwell. All requirements must hold on the
same trajectory. Starting at $x_0=0$, the three-target specification is
$ \psi_3=F_{[1,1.4]}G_{[0,2]}(p_1\ge4) \wedge F_{[1,1.5]}G_{[0,2]}(p_2\ge4.6) \wedge
F_{[1,1.6]}G_{[0,2]}(p_3\ge5.2)$. The two-target case uses the first two
conjuncts. The six-target case adds three requirements with bounds
$p_1\ge5.8$, $p_2\ge6.4$, and $p_3\ge7$, retaining the same respective
temporal intervals. We allow 64 attempts and up to $3000\text{ s}$
extension per $F$, in Algorithm~\ref{alg:temporal_repair}. Lower
endpoints remain fixed and $G$ contraction is permitted.

\begin{table}[t]
  \caption{Geometric $F/G$ update counts and interval changes
    (maximum $F$ extension / total $G$ reduction, s), with median
    algorithm times (ms). MILP uses $\Delta t=0.1$ s;
    $\dagger$ marks acceptance of the original sampled formula.}
  \label{tab:iterative_repair}
  \centering
  {\footnotesize\setlength{\tabcolsep}{3pt}
  \begin{tabular}{l c r r c}
    \toprule
    Case & $F/G$ & Geom. & MILP & Changes (s) \\
    \midrule
    Two targets & $2/0$ & $0.421$ & $519.1$ & $6.167/0$ \\
    Three targets & $3/0$ & $0.632$ & $1051.4$ & $6.865/0$ \\
    Six targets & $6/0$ & $1.510$ & $3125.4$ & $9.013/0$ \\
    Safety fallback & $1/1$ & $0.530$ & $165.3^\dagger$ & $6.940/0.188$ \\
    \bottomrule
  \end{tabular}}
\end{table}

For the proposed geometric decision procedure, the multi-target cases
need two, three, and six $F$ updates, respectively, preserving all
dwells. Each update follows the largest eligible weight; improvements
are retained and weights recomputed
(Lemma~\ref{lem:repair_monotonicity}). With an attempt budget $K=1$,
the two-target case returns \texttt{repair not found}.

For safety fallback, $x_0=(0.996,0,0,0.17,0,0)^T$ and
$\psi=G_{[0,0.2]}(p_1\le1)\wedge F_{[1,3]}G_{[0,2]}(p_2\ge6)$. After
$F$-only repair fails, one $G$ contraction ends safety at
$0.0119383\text{ s}$; the arrival deadline becomes $9.93951\text{ s}$.
The arrival dwell stays two seconds and continuous robustness is
$0.0021843\text{ m}$. MILP at $\Delta t=0.1\text{ s}$ accepts the original
sampled formula, but its trajectory violates continuous safety.

Two positive branch values of $0.02\text{ m}$ also test smoothing-only
rejection: their LSE conjunction is $-0.04931\text{ m}$, requiring two
$F$ updates. Across these cases and the main UAV example, all 18
geometric runs satisfy continuous STL. Median
runtimes are $0.265$--$1.510\text{ ms}$, including all repair checks
(Theorem~\ref{thm:repair_complexity}). MILP completes 30 of 54 runs: 26
witnesses violate continuous STL, two have positive robustness, and two
are within numerical tolerance of zero.

\subsection{Phase 2: Horizon Scalability and Syntactic Complexity}
\label{sec:exp_phase2}

\subsubsection*{Part 1: Horizon Length and Discretization Scalability}

\paragraph{Experimental Setup.}
We test every combination of \\ $T\in\{10,50,100,500,1000\}\text{ s}$ and
$\Delta t\in\{0.5,0.1,0.02,0.005\}\text{ s}$. The original formula is
$F_{[1,3]}G_{[0,T-3]}P$: total horizon grows through dwell while the
arrival deadline stays fixed. Each total horizon has three geometric
repetitions; each of the 20 horizon--sampling combinations has three MILP
repetitions, with the shared resource limits described before.

\paragraph{Results and Temporal Complexity Verification.}
All 15 geometric runs finish, with median runtimes
$0.245$--$0.280\text{ ms}$ for complete repair and
$0.0587$--$0.0643\text{ ms}$ for separately measured compiled queries.
Figure~\ref{fig:phase2_scalability} plots geometry at the corresponding
MILP horizon segment counts up to $N=200{,}000$, its measurements
require no sampling period. At $\Delta t=0.5\text{ s}$, MILP completes
$T=10,50,100,500\text{ s}$. At $\Delta t=0.1\text{ s}$, $T=10$ and
$50\text{ s}$ complete in $3.379$ and $39.963\text{ s}$. No finer-grid
configuration completes the full pipeline.

Of 60 MILP repetitions, 18 complete, 29 exhaust the wall budget, six
reach the memory guard, five leave IIS diagnosis unresolved, and two
leave repaired-model verification unresolved. In
Figure~\ref{fig:phase2_scalability}, crosses pool unfinished statuses at
the budget reference height; they are not completion-time measurements.
Theorem~\ref{thm:constant_time} establishes horizon independence for
compiled checks seen in the figure for the proposed geometric decision
procedure.

\subsubsection*{Part 2: Syntactic Complexity and Nested Formulae}

\paragraph{Experimental Setup.}
We vary nested $F$ count over $p\in\{1,2,4,8,\\ 16,32\}$ and distinct
predicate count over $k'\in\{2,5,10,20,50\}$, with one final
$G_{[0,8]}$. Each $F$ interval is $[1/p,3/p]$, keeping the standard STL
arrival interval $[1,3]$ and total horizon fixed. 
Although the nested $F$ chain is equivalent to $F_{[1,3]}$, the MILP
baseline retains its uncollapsed syntax to measure encoding costs.
MILP uses $\Delta t=1/32\text{ s}$. Both methods have three repetitions for
each of the 30 configurations.

\paragraph{Results and Syntactic Complexity Verification.}
All 90 geometric decision procedure runs complete, with runtime medians
$0.228$--$1.310\text{ ms}$; the separately measured compiled queries take
$0.0548$--$0.2711\text{ ms}$. Compilation traverses $p+k'+2$ syntax-tree
nodes, combines consecutive flow maps, and retains each temporal
occurrence's fractional inversion contribution.
Figure~\ref{fig:phase3_nesting} shows the full algorithm cost, whose
growth with formula size includes compilation, weight propagation, and
repair checks.

Only $(p,k')=(1,2)$ completes all three MILP repetitions, with median
runtime $78.054\text{ s}$. The other 87 remain unfinished: 45 at the
original solve, 25 at IIS diagnosis, 15 at the wall budget, and two at
optimal slack computation. Their crosses indicate status at the
reference height, rather than measured completion times.

\subsection{Phase 3: Soundness and Conservative Rejections}
\label{sec:exp_phase3}

\paragraph{Monte Carlo Empirical Verification.}
We sample $1000$ initial states with seed 42 and reuse this population
for every $\eta\in\{1,2,5,10,20,50\}$, running the proposed technique.
Positions are independently uniform over
$[10,100]\times[10,100]\times[5,95]\text{ m}$ and velocities over
$[-0.5,2]^3\\ \text{ m/s}$, the benchmark's controller-valid domain. For
the Phase~1 specification in Section~\ref{sec:exp_phase1}, the independent physical reference
is $m_{\mathrm{phys}}(x_0)=\min_j\{p_{0,j}+D(3;v_{0,j})-q_j\}$. It
identifies 895 feasible and 105 infeasible states; all 1000 receive a
geometric classification at every $\eta$.

To separate smoothing from temporal-inversion conservatism, we use
$z_j(x_0)=p_{0,j}+D(1;v_{0,j})-q_j$ in
Equation~\eqref{eq:experimental_fractional_field}. On the sampled
velocity range, $D(1;v)=v(1\text{ s})+1.5\text{ m}$. Define the local
LSE loss $\ell_\eta(x_0)=\min_j z_j(x_0)-S_\eta(z(x_0))$.
Lemma~\ref{lem:lse_conservatism} gives
$0\le\ell_\eta(x_0)\le\delta_\eta=\ln(3)/\eta$, with $\eta$ in
$\text{m}^{-1}$. 

\begin{figure}[t]
  \centering
  \includegraphics[width=0.95\columnwidth]{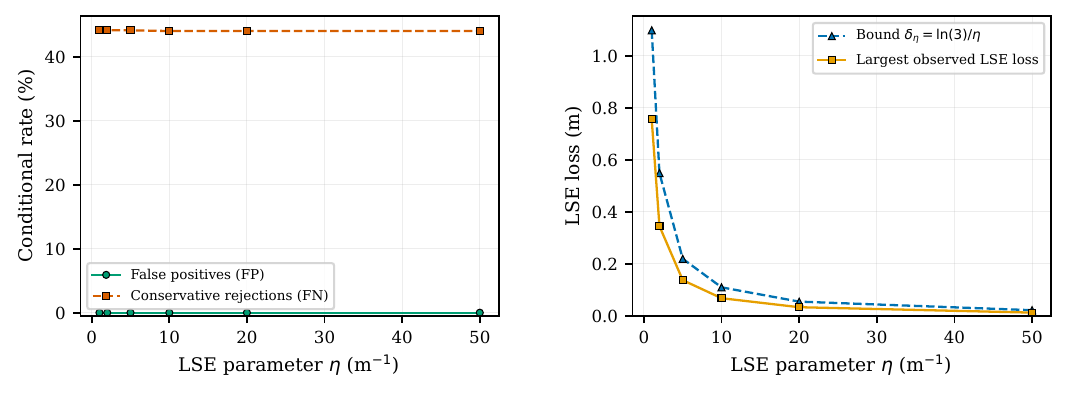}
  \caption{Validation across 1000 paired 6D states as $\eta$ varies.
    Left: conditional false-positive and conservative-rejection rates.
    Right: the LSE bound $\delta_\eta=\ln(3)/\eta$ and the largest observed
    smoothing loss $\max_{x_0}\ell_\eta(x_0)$ over the sampled states.}
  \label{fig:phase3_verification}
\end{figure}

\paragraph{Results Analysis.}
Figure~\ref{fig:phase3_verification} (left) reports false positive rate
$\mathrm{FPR}$ and false negative rate $\mathrm{FNR}$, as percentages of
physically infeasible and feasible states, respectively.
There are zero observed false positives for every $\eta$, consistent with
Theorem~\ref{thm:geometric_soundness} on the tested domain.
False-negative rates are $44.134\%$ for $\eta=1,2,5$ and $44.022\%$ for
$\eta=10,20,50$. The rejections include temporal-inversion conservatism,
rather than only the LSE smoothing strip of
Lemma~\ref{lem:lse_conservatism}. Every observed LSE loss in the right
panel satisfies that lemma's bound. Increasing $\eta$ therefore reduces
smoothing loss while 394 physically feasible states still fail the
temporal certificate. These results illustrate
Proposition~\ref{prop:incompleteness} and separate its quantified
Boolean contribution from the additional conservatism of the temporal
inversions.

\section{Related Work}
\label{sec:related_work}

Our geometric decision procedure sits at the intersection of temporal
logic synthesis, runtime verification, fault diagnosis, and continuous
satisfiability checking.

\textbf{STL Synthesis and Control:} Standard Signal Temporal Logic (STL)
control synthesis encodes continuous dynamics into Mixed-Integer Linear
Programs (MILPs)~\cite{raman2014synthesis,
  belta2019formal,ghosh2016diagnosis}. However, requiring binary
variables at every discrete step causes computational complexity to
scale multiplicatively with horizon length~\cite{kurtz2022mixed}. Our
geometric approach evaluates temporal constraints as continuous spatial
boundaries at time zero, with horizon-independent checks for fixed
compiled dimensions. This supports deployment as both an offline MILP
pre-computation screen, a real-time feasibility gate for Receding
Horizon Control (RHC), or for a LLM in the loop planning.

\textbf{Runtime Verification:} When offline synthesis is computationally
prohibitive, runtime tools like Breach~\cite{donze2010breach} and
RTAMT~\cite{yamaguchi2024rtamt} compute quantitative robustness to
evaluate ongoing trace satisfaction. Diverging from trace monitoring,
our approach certifies synthesis feasibility from a given initial
state prior to execution and prescribes actionable temporal repairs
rather than exclusively providing post-hoc feedback.

\textbf{Diagnosis, Repair, and Infeasible Sets:} As automated natural
language-to-STL translation pipelines emerge~\cite{he2022deepstl},
resolving synthesized conflicting specifications is critical. Existing
diagnosis tools extract Irreducible Infeasible Sets
(IIS)~\cite{ghosh2016diagnosis}, generate
counterexamples~\cite{kress2018synthesis}, or extract tableau-based
unsatisfiable cores~\cite{zamponi2026stlsat}.
STLSat~\cite{zamponi2026stlsat} checks specification satisfiability for
bounded discrete-time STL, while Kress-Gazit et
al.~\cite{kress2018synthesis} survey robot controller synthesis and
diagnostic feedback. Our focus is continuous-time control feasibility
from a given initial state. Our fast feasibility checking and temporal
repair can support online recovery, while control synthesis remains
delegated to the execution controller. To guarantee physical limits,
Hamilton-Jacobi backward reachability computes exact feasible
sets~\cite{chen2018signal} but suffers the curse of dimensionality,
while other approaches adapt spatiotemporal boundaries
online~\cite{basu2025spatiotemporal}. In contrast, our geometric
compilation reduces feasibility checking to a local $\mathcal{O}(nd)$ polyhedral
evaluation. While recent optimization methods tackle Model Predictive
Control (MPC) feasibility restoration using spatial slack
variables~\cite{wu2026feasibility}, our framework uses maximum compiled
residuals to propose analytical temporal adjustments. 

\textbf{SMT and Formal Verification:} Tools like
dReal~\cite{gao2013dreal} implement $\delta$-complete SMT solving for
nonlinear constraints, while reachability engines such as
SpaceEx~\cite{frehse2011spaceex}, Flow*~\cite{chen2013flowstar}, and
Ariadne~\cite{benvenuti2014ariadne} rigorously compute reachable sets
over entire state spaces. These orthogonal formal methods prioritize
exhaustive reachability, whereas our framework provides rapid,
point-wise feasibility checks and finite-time temporal repairs tailored
specifically for pre-deployment screening and automated control loops.

\section{Conclusion and Future Work}
\label{sec:conclusion}
We introduced a geometric procedure for continuous STL feasibility and
temporal repair under actuator limits. Time-zero spatial compilation
allows local certificate evaluation without a temporal sampling loop.
The maximum deficit supplies an analytical first repair; iteration
recompiles after further liveness extensions and, when permitted
dwell-preserving repairs fail, safety contractions. Under the joint
controller assumption, an accepted complete certificate guarantees an
admissible control with nonnegative continuous STL robustness. Compiled
checks are horizon-independent. Benchmarks show lower runtimes than
MILP, including for uncollapsed nested formulae. Future work will reduce
conservatism and extend the approach to nonlinear dynamics.

\bibliographystyle{ACM-Reference-Format}
\bibliography{ref}%

\end{document}